\documentclass[11pt]{article}
\usepackage{fullpage}
\usepackage{amsmath,amsthm,amsfonts,dsfont}
\usepackage{amssymb,latexsym,graphicx}
\usepackage{mathtools}
\usepackage{hyperref}
\usepackage{xcolor}
\hypersetup{
	colorlinks,
	linkcolor={red!70!black},
	citecolor={blue!70!black},
	urlcolor={blue!70!black}
}

\newtheorem{theorem}{Theorem}[section]

\newtheorem*{remark*}{Remark}

\newtheorem{lemma}[theorem]{Lemma}

\newtheorem{corollary}[theorem]{Corollary}
\newtheorem*{definition*}{Definition}
\newtheorem{definition}[theorem]{Definition}

\newcommand{\F}{\ensuremath{\mathbb{F}}}
\newcommand{\R}{\ensuremath{\mathbb{R}}}
\newcommand{\Z}{\ensuremath{\mathbb{Z}}}

\newcommand{\lat}{\mathcal{L}}

\newcommand{\eps}{\varepsilon}
\renewcommand{\epsilon}{\varepsilon}
\newcommand{\poly}{\mathrm{poly}}

\renewcommand{\vec}[1]{\ensuremath{\boldsymbol{#1}}}

\DeclarePairedDelimiter\floor{\lfloor}{\rfloor}
\DeclarePairedDelimiter\ceil{\lceil}{\rceil}

\mathtoolsset{centercolon}

\usepackage[hyperpageref]{backref}

\newcommand{\ksum}[1]{#1\text{-}\mathrm{SUM}}
\def\kSUM{\ksum{k}}

\newcommand{\kSUMG}[3]{#1\text{-}\mathrm{SUM}({#2},{#3})}

\begin{document}
\title{On the Hardness of Average-case $k$-SUM}
\author{
    Zvika Brakerski\thanks{Supported by the Binational Science Foundation (Grant No.\ 2016726), and by the European Union Horizon 2020 Research and Innovation Program via ERC Project REACT (Grant 756482) and via Project PROMETHEUS (Grant 780701).}\\Weizmann Institute of Science\\
    \texttt{zvika.brakerski@weizmann.ac.il} \and
		Noah Stephens-Davidowitz\footnotemark[2]\\Cornell University\\
		\texttt{noahsd@gmail.com} \and
    Vinod Vaikuntanathan\thanks{Supported in part by NSF Grants CNS-1350619, CNS-1414119 and CNS-1718161, Microsoft Faculty Fellowship and an MIT/IBM grant.}\\ MIT\\
    \texttt{vinodv@csail.mit.edu}
}

\maketitle

\begin{abstract}
In this work, we show the first worst-case to average-case reduction for the classical $k$-SUM problem. A $k$-SUM instance is a collection of $m$ integers, and the goal of the $k$-SUM problem is to find a subset of $k$ elements that sums to $0$. In the average-case version, the $m$ elements are chosen uniformly at random from some interval $[-u,u]$.

We consider the \emph{total} setting where $m$ is sufficiently large (with respect to $u$ and $k$), so that we are guaranteed (with high probability) that solutions must exist. Much of the appeal of $k$-SUM, in particular connections to problems in computational geometry, extends to the total setting.

The best known algorithm in the average-case total setting is due to Wagner (following the approach of Blum-Kalai-Wasserman), and achieves a run-time of $u^{O(1/\log k)}$. This beats the known (conditional) lower bounds for worst-case $\kSUM$, raising the natural question of whether it can be improved even further.  However, in this work, we show a matching {\em average-case lower-bound}, by showing a reduction from \emph{worst-case lattice problems}, thus introducing a new family of techniques into the field of fine-grained complexity. In particular, we show that any algorithm solving average-case $k$-SUM on $m$ elements in time $u^{o(1/\log k)}$ will give a super-polynomial improvement in the complexity of algorithms for lattice problems.
\end{abstract}

\thispagestyle{empty}
\newpage
\tableofcontents
\pagenumbering{roman}
\newpage

\pagenumbering{arabic}

\section{Introduction}
\label{sec:intro}

The $k$-SUM problem is a parameterized version of the classical subset sum problem. 
Given a collection of $m$ integers $a_1,\ldots,a_m$, the $k$-SUM problem asks if there is some subset of cardinality $k$ that sums to zero.\footnote{This is the homogeneous version of $k$-SUM. One could also define the inhomogeneous version where the instance consists also of a target integer $t$, and the goal is to produce a subset of $k$ elements that sums to $t$. In the worst-case world, the two versions are equivalent.}
This problem (especially for $k=3$, but more generally for arbitrary constant $k$)
has been influential in computational geometry,
where reductions from $\kSUM$ have been used to show the conditional
hardness of a large class of problems~\cite{GOClassProblems95,GOClassProblems12}.
More generally it has been used in computational complexity, where it has formed
the basis for several fine-grained hardness results~\cite{P10,AV14,VW13,KPP16}.  We refer the reader to the extensive
survey of Vassilevska-Williams~\cite{WilFinegrainedQuestions18} for an exposition of this line of work. The $k$-SUM problem has also been extensively studied in the cryptanalysis community (see, e.g., \cite{Wagner02,BKW03,BCJ11}). 

We know two very different algorithms for $k$-SUM: a meet-in-the-middle algorithm that achieves run-time $O(m^{\lceil k/2\rceil})$~\cite{HS74}, and dynamic programming or FFT-based algorithms that achieve run-time $\widetilde{O}(um)$~\cite{Bellman}  where $u$ is the largest absolute value of the integers $a_i$ (A sequence of recent works~\cite{KX17,Bringmann17,ABJTW19,JW19}  improve the latter to $\widetilde{O}(u+m)$). Note that the latter algorithms outperform the former when $u \ll m^{\lceil k/2\rceil}$, in what is sometimes called the {\em dense regime} of parameters, a point that we will come back to shortly.

In terms of hardness results for $k$-SUM, the work of P\v{a}tra\c{s}cu and Williams~\cite{PW10} shows that an algorithm that solves the problem in time $m^{o(k)}$ for all $m$ will give us better algorithms for
SAT, in particular refuting the exponential time hypothesis (ETH). The recent work of Abboud, Bringmann, Hermelin and Shabtay~\cite{ABHS19} shows that a $u^{1-\epsilon}$-time algorithm (for any constant $\epsilon>0$) would refute the strong exponential-time hypothesis (SETH).  So, we know that the two algorithms described above are essentially optimal, at least in the worst case.

\paragraph{Average-case Hardness.} 
The focus of this work is the natural average-case version of $k$-SUM where the problem instance $a_1,\ldots,a_m$ is chosen independently and uniformly at random from an interval $[-u,u]$. We call this the {\em average-case $\kSUM$ problem}. In this setting, \emph{deciding} whether a $\kSUM$ solution exists is in many cases trivial. In particular, if ${m \choose k} \ll u$ then a union bound argument shows that the probability of a solution existing approaches $0$. We refer to this as the \emph{sparse} regime of the problem. In contrast, if ${m \choose k}$ is sufficiently larger than $u$, then a hashing argument guarantees the existence of many solutions, with high probability over the instance. As already mentioned above, we refer to this as the \emph{dense} regime.

Notwithstanding this triviality, we notice that in the dense regime one could still consider the {\em search} problem of \emph{finding} a $\kSUM$ solution. The search problem seems to retain its hardness even in the dense setting and is the focus of our work. Since we consider the search version of the problem, we also refer to the dense regime as the \emph{total} regime, as the associated search problem has a solution with high probability.

The average-case total problem is {\em not} quite as hard as the worst-case version (at least assuming SETH), since (slight variants of) Wagner's generalized birthday algorithm~\cite{Wagner02} and the Blum-Kalai-Wasserman algorithm~\cite{BKW03} show how to solve this problem in time $u^{O(1/\log k)}$. This contrasts with the $u^{1-\eps}$ lower bound of~\cite{ABHS19} in the worst case. (The BKW/Wagner algorithm was originally stated in a slightly different setting, so we restate it in Section~\ref{sec:BKW}.) This leaves the question of \emph{how much easier} the average-case is compared to the worst-case. Given that the lower-bounds from the worst-case setting are not a barrier here, it is a-priori unclear what is the best running time in this setting. Can we improve on \cite{BKW03,Wagner02}?

\subsection{Our Results}

In this work we characterize the hardness of average-case $\kSUM$ in the total regime by presenting a (conditional) lower bound that matches the $u^{O(1/\log k)}$ upper bound described above, up to the hidden constant in the exponent. 

In more detail, our main result shows that average-case $\kSUM$ is indeed hard to solve, under the assumption that \emph{worst-case} lattice problems are hard to approximate. We thus introduce a new family of techniques into the study of the hardness of the $\kSUM$ problem.  Concretely, this lower-bound shows that a $u^{o(1/\log k)}$-time algorithm for average-case $k$-SUM (in the dense regime) implies a $2^{o(n)}$-time $n^{1+\epsilon}$-approximation algorithm for the shortest independent vectors problem (SIVP) over an $n$-dimensional lattice, a lattice problem for which the best known algorithms run in time $2^{\Omega(n)}$~\cite{ALNSSlideReduction20,ACNoteConcrete19}.   Improving this state of affairs, in particular finding a $2^{o(n)}$-time algorithm for SIVP, would have major consequences in lattice-based cryptography both in theory and in practice~\cite{NIST,APSConcreteHardness15, ADTSPostquantumKey16}.

We also note in the appendix that some of the connections between $\kSUM$ and geometric problems from \cite{GOClassProblems95,GOClassProblems12} carry over to the dense setting as well. This shows an interesting (and not previously known, as far as we could find) connection between approximate short vectors in lattices, and computational geometry.

\def\vecx{\mathbf{x}}

\subsection{Our Techniques}

The starting point of our reduction is the well-known worst-case to average-case reductions in the lattice world, pioneered by Ajtai~\cite{AjtGeneratingHard96,MRWorstcaseAveragecase07,GPVTrapdoorsHard08,GINX16}. These reductions show that the approximate shortest independent vectors (SIVP) problem, a standard problem in the lattice world, is at least as hard in the worst-case as a certain problem called short integer solutions (SIS) on the average. The definition of lattices and the approximate shortest vector problem is not crucial for the current discussion, however we note that the best algorithms on $n$-dimensional lattices that compute any $\poly(n)$-approximation to SIVP run in time $2^{\Omega(n)}$. 
(We refer the curious reader to, e.g., \cite{MRWorstcaseAveragecase07,PeiDecadeLattice16, ALNSSlideReduction20}, Section~\ref{sec:lattices}, and the references therein for more background on lattices and lattice problems.)

In the (one-dimensional) {\em average-case} SIS problem with parameters $m,Q$ and $\beta$, one is given random integers $a_1,\ldots,a_m \in \Z_Q$ and the goal is to find a {\em non-zero} integer linear combination $\vecx = (x_1,\ldots,x_m) \in \Z^m$ such that $\sum_{i\in [m]} a_ix_i = 0 \pmod{Q}$ and $\vecx$ is short, namely $||\vecx||_1 \leq \beta$. Thus, this is exactly the modular subset sum problem (i.e.\ subset sum over the group $\Z_Q$), except with weights larger than $1$. The parameters of the problem live in the dense/total regime where such solutions are guaranteed to exist with high probability. 
The worst-case to average-case reductions state that an average-case SIS solver for a sufficently large $Q$, namely $Q = (\beta n)^{\Omega(n)}$, gives us an $\widetilde{O}(\sqrt{n\log m}\cdot \beta)$-approximate algorithm for SIVP. (We refer the reader to Theorem~\ref{thm:GINX} for a more precise statement.)

Our main technical contribution is an average-case to average-case reduction from the SIS problem to the $k$-SUM problem. We show this by exhibiting a reduction from SIS to modular $k$-SUM (i.e.\ $\kSUM$ over the group $\Z_Q$), and one from modular $k$-SUM to $k$-SUM. The latter is easy to see: indeed, if you have a $k$-subset that sums to $0$, it also sums to $0 \pmod{Q}$ for any $Q$.
Henceforth in this discussion, when we say $k$-SUM, we will mean modular $k$-SUM.

To reduce from SIS with parameters $m,Q,\beta$ to modular $k$-SUM on $m$ numbers over $\Z_Q$, we start with a simple, seemingly {\em trivial}, idea. SIS and $k$-SUM are so similar that perhaps one could simply run the $k$-SUM algorithm on the SIS instance. Unfortunately, this fails. For a $k$-SUM solution to exist, $m$ has to be at least roughly $Q^{1/k} = n^{\Omega(n/k)}$. But, this could only possibly give us an approximate-SIVP algorithm that runs in time $n^{\Omega(n/k)}$ (where we are most interested in constant $k$), since the reduction from SIVP has to at least write down the $m$ samples. This is   a meaningless outcome since, 
as we discussed before, there are algorithms for approximate SIVP that run in time $2^{O(n)}$.

Fortunately, ideas from the BKW algorithm~\cite{BKW03} for subset sum (and the closely related algorithm from \cite{Wagner02} for $k$-SUM) come to our rescue. We will start with SIS modulo $Q = q^L$ for some $q$ and $L$ that we will choose later.  (\cite{GINX16} showed that worst-case to average-case reductions work for any sufficiently large $Q$, include $Q = q^L$.)

The BKW algorithm iteratively produces subsets that sum to $0$ modulo $q^i$ for $i=1,\ldots,L$, finally producing SIS solutions modulo $Q$. To begin with, observe that for a $k$-subset-sum to exist modulo $q$, it suffices that $m \approx q^{1/k} \ll Q^{1/k}$, potentially getting us out of the conundrum from before. In particular, we will set $q \approx 2^{n\log k}$, $L \approx \log n/\log k$, therefore $Q = q^L \approx n^n$ as needed. We will also set  $m \approx q^{\epsilon/\log k} \approx 2^{\epsilon n}$ for a large enough $\epsilon$ so that solutions exist (since $m^k \gg q$). Furthermore, a $k$-SUM algorithm mod $q$ that performs better than BKW/Wagner, that is, runs in time $q^{o(1/\log k)} =  2^{o(n)}$, {\em is} potentially useful to us. 

With this ray of optimism, let us assume that we can run the $k$-SUM algorithm many times to get several, $m$ many, subsets $S_j$ that sum to $0$ modulo $q$.  (We will return to, and remove, this unrealistic assumption soon.) 
That is, 
$$ b_j := \sum_{i \in S_j} a_i = 0 \pmod{q}$$
The BKW/Wagner approach would then be to use the $(b_1,\ldots,b_m)$ to generate $(c_1,\ldots,c_m)$ that are $0 \pmod{q^2}$, and so on. Note that $c_i$ are a linear combination of $a_1,\ldots,a_m$ with weight $k^2$. At the end of the iterations, we will obtain at least one linear combination of $(a_1,\ldots,a_m)$ of  weight $\beta = k^L$ that sums to $0$ modulo $q^L = Q$, solving SIS. (We also need to make sure that this is a non-zero linear combination, which follows since the coefficients of all intermediate linear combinations are positive.)

This would finish the reduction, except that we need to remove our unrealistic assumption that we can use the $k$-SUM oracle to get many $k$-subsets of $(a_1,\ldots,a_m)$ that sum to $0$. For one, the assumption is unrealistic because if we feed the $k$-SUM oracle with the same $(a_1,\ldots,a_m) \pmod{q}$ twice, we will likely get the same $k$-SUM solution.
On the other hand, using a fresh random instance for every invocation of the $k$-SUM oracle will require $m$ to be too large (essentially returning to the trivial idea above). A natural idea is to observe that each $k$-SUM solution touches a very small part of the instance. Therefore, one could hope to first receive a $k$-SUM $a_{i_1}+\ldots+a_{i_k}$ from the oracle, and in the next iteration, use as input $\{a_1,\ldots, a_m\} \setminus \{a_{i_1},\ldots,a_{i_k}\}$, which is nearly as large as the original set. Unfortunately, continuing like this cannot work. The distributions of the successive instances that we feed to the oracle will no longer be uniform, and even worse, the oracle itself can choose which elements to remove from our set. A suitable malicious oracle can therefore prevent us from obtaining many $k$-SUMs in this way, even if the oracle has high success probability on uniform input.

Instead, our key idea is rather simple, namely to resort to randomization.
Given an instance $(a_1,\ldots,a_m) \in \Z_q^m$, we compute many random subset sums to generate $(a_1',\ldots,a_{m}') \in \Z_q^{m}$. That is, we choose  $k$-subsets $T'_i \subseteq [m]$ and let 
$$ a_i' = \sum_{j \in T'_i} a_j \pmod{q}$$
Since $q \gg m^{1/k}$, the leftover hash lemma~\cite{ILL89} tells us that the $a_i'$ are (statistically close to) uniformly random mod $q$. Furthermore, a $k$-subset sum of $(a_1',\ldots,a_m')$ will give us a $k^2$-subset sum of $(a_1,\ldots,a_m)$ that sums to $0 \pmod{q}$. To obtain a new subset sum of $(a_1,\ldots,a_m)$, simply run this process again choosing fresh subsets $T_i''$ to generate $(a_1'',\ldots,a_m'')$; and so on.  Eventually, this will give us a $\beta = k^{2L}$ weight solution to SIS, which is a quadratic factor worse than before, but good enough for us. (We are glossing over an important technical detail here, which is how we ensure that the resulting subset sums yield uniformly random independent elements in $q\Z/q^2\Z$.)

To finish the analysis of the reduction, observe that it calls the $k$-SUM oracle $\approx mL$ times. Assuming the oracle runs in time $q^{o(1/\log k)}$, this gives us a $2^{o(n)}$-time algorithm for approximate SIVP. The approximation factor is $\widetilde{O}(\sqrt{n\log m}\cdot \beta) \approx n^{3}$. (In the sequel, we achieve $n^{1+o(1)}$ by a careful choice of parameters.)

Interestingly, our reduction re-imagines the BKW/Wagner {\em algorithm as a reduction} from the SIS problem to $k$-SUM, where the algorithm itself is achieved (in retrospect) by plugging in the trivial algorithm for $\ksum{2}$. Of course, eventually the algorithm ends up being much simpler than the reduction (in particular, there is no need for re-randomization) since we don't need to account for ``malicious'' $\kSUM$ solvers.

\subsection{Open Problems and Future Directions}

Our work introduces the powerful toolkit of lattice problems into the field of average-case fine-grained complexity, and raises several natural directions for further research. 

First is the question of whether a result analogous to what we show holds in the sparse/planted regime as well. A possible theorem here would rule out an $m^{o(k)}$-time algorithm for $k$-SUM, assuming the hardness of lattice problems. To the best of our knowledge, in the sparse/planted regime it is not known whether the average-case problem is easier than the worst-case as in the dense regime. 

Second is the question of whether we can obtain average-case hardness of $k$-SUM for concrete small constants~$k$, perhaps even $k=3$. Our hardness result is asymptotic in $k$. 

Third is the question of whether we can show the average-case hardness of natural distributions over {\em combinatorial and computational-geometric} problems, given their connection to $k$-SUM. In this vein, we show a simple reduction to (perhaps not the most natural distribution on) the $(Q,m,d)$-plane problem in Appendix~\ref{apx:geometric}, but we believe much more can be said. More generally, now that we have shown average-case hardness of $k$-SUM, it is natural to try to reduce average-case $k$-SUM to other natural average-case problems.

\subsection{Other Related Works}

There are now quite a few works that study average-case fine-grained hardness of problems in $P$. We mention a few. First, Ball, Rosen, Sabin, and Vasudevan~\cite{BallRSV17} showed a reduction from SAT to an (average-case) variant of the orthogonal vectors problem. They demonstrated that sub-quadratic algorithms for their problem will refute SETH. 

There is also a sequence of works on the average-case hardness of counting $k$-cliques. The work of Goldreich and Rothblum~\cite{GoldreichR18,GR20} shows worst-case to average-case {\em self}-reductions for the problem of counting $k$-cliques (and other problems in $P$). Boix-Adser{\`a}, Brennan, and Bresler proved the same result for $G_{n,p}$~\cite{boix-adseraAverageCaseComplexityCounting2019}, and Hirahara and Shimizu recently showed that it is even hard to count the number of $k$-cliques with even a small probability of success~\cite{hiraharaNearlyOptimalAverageCase2020}. In contrast, our reductions go from the worst-case of one problem (SIVP) to the average-case of another ($k$-SUM). We find it a fascinating problem to show a worst-case to average-case {\em self}-reduction for $k$-SUM. 

Dalirrooyfard, Lincoln, and Vassilevska Williams recently proved fine-grained average-case hardness for many different~\cite{DLW20} problems in $P$ under various complexity-theoretic assumptions. In particular, they show fine-grained average-case hardness of counting the number of solutions of a ``factored'' variant of $k$-SUM assuming SETH. In contrast, we show fine-grained average-case hardness of the standard $k$-SUM problem under a non-standard assumption.

For the lattice expert, we remark that if one unwraps our reduction from SIVP to SIS and then to $k$-SUM, we obtain a structure that is superficially similar to \cite{MPHardnessSIS13}. However, in their setting, they do not need to reuse samples and therefore do not need the re-randomization technique, which is the key new idea in this work. 

\subsection{Organization of the Paper} 

Section~\ref{sec:variants} describes the modular variant of $k$-SUM as well as the standard $k$-SUM (over the integers), shows their totality on average, and reductions between them. For completeness, we describe the BKW/Wagner algorithm in Section~\ref{sec:BKW}. We remark that while the standard descriptions of the algorithm refer to finite groups, we need one additional trick (namely, Lemma~\ref{lem:Qimpliesu}) to obtain the algorithm over the integers. Finally, our main result, the worst-case to average-case reduction is described in Section~\ref{sec:wcac}. The connection to computational geometry is provided in Appendix~\ref{apx:geometric}.

\section{Preliminaries}

We write $\log$ for the logarithm base two and $\ln$ for the natural logarithm. We write $\binom{m}{k} := \frac{m!}{(m-k)!k!}$ for the binomial coefficient. 

\subsection{Probability}

We make little to no distinction between random variables and their associated distributions.
For two random variables $X, Y$ over some set $S$, we write $\Delta(X,Y) := \sum_{z \in S} |\Pr[X = z] - \Pr[Y=z]|$ for the statistical distance between $X$ and $Y$. For a finite set $S$, we write $U_S$ for the uniform distribution over $S$.

Recall that a set of functions $\mathcal{H} \subseteq \{h : X \to Y\}$ is a \emph{universal family of hash functions} from $X$ to $Y$ if for any distinct $x,x' \in X$
\[
    \Pr_{h \sim \mathcal{H}}[h(x) = h(x')] \leq 1/|Y|
    \; .
\]

\begin{lemma}[Leftover hash lemma]
    If $\mathcal{H}$ is a universal family of hash functions from $X$ to $Y$, then
    \[
        \Pr[\Delta(h(U_X), U_Y) \geq \beta] \leq \beta
        \; ,
    \]
    where the probability is over a random choice of $h \sim \mathcal{H}$ and $\beta := (|Y|/|X|)^{1/4}$.
\end{lemma}

\begin{lemma}
    For any positive integers $Q,m$, let $\mathcal{H}$ be the family of hash functions from $\{0,1\}^m$ to $\Z_Q$ given by $h_{\vec{a}}(\vec{x}) = \langle \vec{a},\vec{x} \rangle \bmod Q$ for all $\vec{a} \in \Z_Q^{m}$. Then, $\mathcal{H}$ is a universal family of hash functions.
\end{lemma}
\begin{proof}
    Let $\vec{x}, \vec{y} \in \{0,1\}^m$ be distinct vectors, and suppose without loss of generality that $x_1 =1$ and $y_1 = 0$. We write $\vec{a}' \in \Z_Q^{m-1}$ for the vector obtained by removing the first coordinate from $\vec{a}$ and ${a}_1$ for the first coordinate itself. Similarly, we write $\vec{x}', \vec{y}' \in \{0,1\}^{m - 1}$ for the vectors $\vec{x}, \vec{y}$ with their first coordinate removed. Then,
    \[
        \Pr[\langle \vec{a},\vec{x}\rangle = \langle\vec{a},\vec{y} \rangle\bmod Q] = \Pr[\langle\vec{a}', \vec{x}'\rangle + {a}_1 = \langle\vec{a}', \vec{y}'\rangle \bmod Q] = \Pr[{a}_1 = \langle\vec{a}',\vec{y}' - \vec{x}'\rangle \bmod Q] = 1/Q
        \; ,
    \]
    where the probability is over the random choice of $\vec{a} \in \Z_Q^m$. The last equality follows from the fact that ${a}_1 \in \Z_Q$ is uniformly random and independent of $\vec{a}'$.
\end{proof}

\begin{corollary}
\label{cor:LHL_A_1}
    For any positive integers $Q,m$ and any subset $X \subseteq \{0,1\}^m$,
    \[
        \Pr_{\vec{a} \sim \Z_Q^{m}}[\Delta(\langle \vec{a}, U_X \rangle \bmod Q, U_{\Z_Q}) \geq \beta] \leq \beta
        \; ,
    \]
    where $\beta := (Q/|X|)^{1/4}$.
\end{corollary}

\begin{corollary}
\label{cor:LHL_A}
    Let $\vec{a} := ({a}_1,\ldots, {a}_M) \in \Z_Q^{M}$ be sampled uniformly at random, and let $S_1,\ldots, S_{M'} \subset [M]$ be sampled independently and uniformly at random with $|S_i| = t$. Let ${c}_i := \sum_{j \in S_i} {a}_j \bmod Q$. Then, $(\vec{a}, \vec{c}) := ({a}_1,\ldots, {a}_M, {c}_1,\ldots, {c}_{M'})$ is within statistical distance $\delta$ of a uniformly random element in $\Z_Q^{M+M'}$, where
    \[
        \delta := (M'+1)\cdot Q^{1/4} \cdot \binom{M}{t}^{-1/4} \leq (M' + 1) \cdot \bigg(\frac{Qt^t}{M^t}\bigg)^{1/4}
    \; .
    \]
\end{corollary}
\begin{proof}
    Let $X_t := \{\vec{x} \in \{0,1\}^M \ : \ \|\vec{x}\|_1 = t\}$, and notice that $|X_t| = \binom{M}{t}$.
    Call $\vec{a}$ good if $\Delta(\langle\vec{a},U_{X_t}\rangle \bmod Q, U_{\Z_Q}) \leq \beta := Q^{1/4} /|X_t|^{1/4}$.
    From Corollary~\ref{cor:LHL_A_1}, we see that $\vec{A}$ is good except with probability at most $\beta$.

    Finally, notice that the ${c}_i$ are distributed exactly as independent samples from $\langle\vec{a}, U_{X_t}\rangle$. Therefore, if $\vec{a}$ is good, each of the ${c}_i$ is within statistical distance $\beta$ of an independent uniform sample. The result then follows from the union bound.
\end{proof}

\subsection{Hitting probabilities}

\begin{definition}
    For $\vec{a} := (a_1,\ldots, a_M) \in \Z_Q^{M}$, $\vec{c} := (c_1,\ldots, c_{M'}) \in \Z_Q^{M'}$, $I \subset [M]$, $J \subset [M']$, and a positive integer $t$, the $t$-\emph{hitting probability} of $\vec{a}$, $\vec{c}$, $I$, and $J$ is defined as follows. For each $j \in J$, sample a uniformly random $S_j \in \binom{[M]}{t}$ with $\sum_{i \in S_j} a_i = c_j$. (If no such $S_j$ exists, then we define the hitting probability to be $1$.) The hitting probability is then
    \[
        p_{\vec{a}, \vec{c}, I, J,t} := \Pr[\exists\; j,j'\in J \text{ such that } S_j \cap I \neq \emptyset \text{ or } S_j \cap S_{j'} \neq \emptyset]
        \; .
    \]
\end{definition}

\begin{lemma}
\label{lem:hitting}
    For any positive integers $Q,M,t$ and $0 < \eps < 1$,
    \[
        \Pr_{\vec{a} \sim \Z_Q^{M}, c \sim \Z_Q}\Big[p_{\vec{a}, c,t} \geq \frac{1+\eps}{1-\eps} \cdot \frac{t}{M}\Big] \leq \frac{4Q^{1/4}}{\eps \cdot \binom{M-1}{t-1}^{1/4}}
        \; .
    \]
    where
    \[
        p_{\vec{a}, c, t} := p_{\vec{a}, c, \{1\}, \{1\},t}
        \; .
    \]
\end{lemma}
\begin{proof}
We have
\[
    p_{\vec{a},c,t} = \Pr_{\vec{x} \sim X_t}[x_1 = 1\ |\ \langle \vec{a}, \vec{x} \rangle = c \bmod Q]
    \; ,
\]
where $X_t := \{\vec{x} \in \{0,1\}^M \ : \ \|\vec{x}\|_1 = t\}$. Therefore,
\begin{align*}
    p_{\vec{a},c,t}
        &= \Pr_{\vec{x} \sim X_t}[x_1 = 1] \cdot \Pr_{\vec{x} \sim X_t}[\langle \vec{a}, \vec{x} \rangle = c \bmod Q\ | \ x_1 = 1]/\Pr_{\vec{x} \sim X_t}[\langle \vec{a}, \vec{x} \rangle = c \bmod Q]\\
        &= \frac{t}{M} \cdot \Pr_{\vec{x}' \sim X_{t-1}'}[\langle \vec{a}_{-1}, \vec{x}'\rangle  = c - a_1 \bmod Q]/\Pr_{\vec{x} \sim X_t}[\langle \vec{a}, \vec{x} \rangle = c \bmod Q]
        \; ,
\end{align*}
where $\vec{a}_{-1}$ is $\vec{a}$ with its first coordinate removed and $X_{t-1}' := \{\vec{x} \in \{0,1\}^{M-1} \ : \ \|\vec{x}\|_1 = t-1\}$.

So, let
    \[
    p_{\vec{a}, c} := \Pr_{\vec{x} \in X_t}[ \langle \vec{a}, \vec{x} \rangle = c \bmod Q ]
    \; ,
    \]
    and
    \[
        p_{\vec{a}, c}' := \Pr_{\vec{x}' \in X_{t-1}'}[ \langle \vec{a}_{-1}, \vec{x}'\rangle = c - a_1 \bmod Q ]
        \; .
    \]
    As in the proof of Corollary~\ref{cor:LHL_A}, we see that
    \begin{equation}
    \label{eq:good}
        \sum_{c \in \Z_Q} |p_{\vec{a},c}- 1/Q| = \Delta(\langle \vec{a}, U_{X_t}\rangle \bmod Q, U_{\Z_Q}) \leq Q^{1/4}/\binom{M}{t}^{1/4}
    \end{equation}
    except with probability at most $Q^{1/4}/\binom{M}{t}^{1/4}$ over $\vec{a}$. Similarly,
    \begin{equation}
    \label{eq:very_good}
        \sum_{c\in \Z_Q} |p_{\vec{a},c}'- 1/Q| = \Delta(\langle \vec{a}_{-1}, U_{X_{t-1}'\rangle } \bmod Q, U_{\Z_Q}) \leq Q^{1/4}/\binom{M-1}{t-1}^{1/4}
    \end{equation}
    except with probability at most $Q^{1/4}/\binom{M-1}{t-1}^{1/4}$ over $\vec{a}$.

    So, suppose that $\vec{a}$ satisfies Eq.~\eqref{eq:good} and Eq.~\eqref{eq:very_good}. Then, by Markov's inequality,
    \[
        \Pr_{c \in \Z_Q}[p_{\vec{a}, c} \geq (1-\eps)/Q] \leq \frac{Q^{1/4}}{\eps \cdot \binom{M}{t}^{1/4}}
    \]
    for any $0 < \eps < 1$, and similarly,
    \[
        \Pr_{c \in \Z_Q}[p_{\vec{a}, c}' \leq (1+r)/Q] \leq \frac{Q^{1/4}}{\eps \cdot \binom{M-1}{t-1}^{1/4}}
        \; .
    \]
    Therefore, for such $\vec{a}$,
    \[
        \Pr[p_{\vec{a},c,t} \geq (1+\eps)t /((1-\eps)M)] \leq \frac{2Q^{1/4}}{\eps \cdot \binom{M-1}{t-1}^{1/4}}
        \; .
    \]
    The result then follows by union bound.
\end{proof}

By repeated applications of union bound, we derive the following corollary.

\begin{corollary}
    \label{cor:hitting_prob}
    For any positive integers $Q,M,M',t, v, v'$ and $0 < \eps < 1$, let $\vec{a} \sim \Z_Q^{M}$ and $\vec{c} \sim \Z_Q^{M'}$ be sampled uniformly at random. Then,
    \[
    p_{\vec{a}, \vec{c},I,J,t} \leq (v+tv') \cdot v' \cdot \frac{1+\eps}{1-\eps} \cdot \frac{t}{M}
    \]
    for all $I \in \binom{[M]}{\leq v},\; J \in \binom{[M']}{\leq v'}$ except with probability at most
    \[
    4 M M'\frac{Q^{1/4}}{\eps \cdot \binom{M-1}{t-1}^{1/4}}
    \; .
    \]
\end{corollary}
\begin{proof}
Let $\eta := \max_{i,j} p_{\vec{a}, \vec{c}, \{i\}, \{j\},t}$.
    By union bound, for any set $I$, we have
    \[
        p_{\vec{a},\vec{c}, I, \{j\},t} \leq \sum_{i\in I} p_{\vec{a},\vec{c},\{i\},\{j\}, t} \leq |I| \cdot \eta
        \; .
    \]
    Fix some set $J$. Let $S_j$ be as in the definition of the hitting probability, and let $I_{-j} := I \cup \bigcup_{j' \in J \setminus \{j\}} S_{j'}$. Notice that $|I_{-j}| \leq |I| + t|J|$.
    Then,
    \[
        p_{\vec{a}, \vec{c}, I, J,t} \leq \sum_{j \in J}  p_{\vec{a},\vec{c}, I_{-j}, \{j\},t} \leq (|I|+t|J|) \cdot |J|\cdot \eta
        \; .
    \]
    Finally, by union bound and Lemma~\ref{lem:hitting}, we have
    \[
        \eta \leq \frac{1+\eps}{1-\eps} \cdot \frac{r}{M}
    \]
    except with probability at most
    \[
    4 M M'\frac{Q^{1/4}}{\eps \cdot \binom{M-1}{t-1}^{1/4}}
    \;.
    \]
    The result follows.
\end{proof}

\subsection{Lattices and Lattice Problems}
\label{sec:lattices}

\begin{definition}[Shortest Independent Vectors Problem]
  \label{def:sivp}
  For an approximation factor $\gamma := \gamma(n) \geq 1$,
  $\gamma$-SIVP is the search problem defined as follows. Given a lattice~$\lat \subset \R^n$, output~$n$ linearly
  independent lattice vectors which all have length at most
  $\gamma(n)$ times the minimum possible, $\lambda_n(\lat)$.
\end{definition}

\def\vecx{\mathbf{x}}

\def\SIS{\mathsf{SIS}}

\begin{definition}[Short Integer Solutions]
     \label{def:sis}
   For integers $m,Q,\alpha$, the (average-case) short integer solutions problem $\SIS(m,Q,\beta)$ is defined by $m$ integers $a_1,\ldots,a_m$ drawn uniformly at random and independently from $\Z_Q$, and the goal is to come up with a {\em non-zero} vector $\vecx = (x_1,\ldots,x_m)$ where $$\sum_{i\in [m]} x_i a_i = 0 \pmod{Q} \hspace{.2in} \mbox{and}\hspace{.2in} ||\vecx||_1 := \sum_{i=1}^m |x_i| \leq \alpha$$
\end{definition}

Following the seminal work of Ajtai~\cite{AjtGeneratingHard96}, there have been several works that show how to solve the {\em worst-case} $\gamma$-SIVP problem given an algorithm for the {\em average-case} SIS problem. We will use the most recent one due to Gama et al.~\cite{GINX16} (specialized to the case of cyclic groups for simplicity).

\begin{theorem}[Worst-Case to Average-Case Reduction for SIS~\cite{MRWorstcaseAveragecase07,GINX16}]\label{thm:GINX}
   Let $n,Q, \beta \in \mathbb{N}$ where $Q = (\beta n)^{\Omega(n)}$. If there is an algorithm for the {\em average-case} SIS problem $\SIS(m,Q,\beta)$ over $\Z_Q$ that runs in time $T$, then there is an $(m+T)\cdot \mathsf{poly}(n)$-time algorithm for {\em worst-case} %
   $\widetilde{O}(\sqrt{n \log m}\cdot  \beta)$-SIVP on any $n$-dimensional lattice $L$.
\end{theorem}

\section{Variants of Average-case \texorpdfstring{$k$}{k}-SUM: Totality and Reductions}\label{sec:variants}

We define two variants of average-case $k$-SUM, one over the integers (which is the standard version of $k$-SUM) and one over the finite group $\Z_Q$ of integers modulo $Q$. We show that the hardness of the two problems is tied together, which will allow us to use the modular version for our results down the line.

\begin{definition}[Average-case $k$-SUM]
For positive integers $m, k \geq 2$ and $u \geq 1$, the {\em average-case} $\kSUMG{k}{u}{m}$ problem is the search problem defined as follows. The input is $a_1,\ldots, a_m \in [-u,u]$ chosen uniformly and independently at random, and the goal is to
find $k$ distinct elements $a_{i_1},\ldots, a_{i_k}$ with $a_{i_1} + \cdots + a_{i_k} = 0$.
\end{definition}

We define the modular version of the problem where the instance consists of numbers chosen at random from the finite additive group $\Z_Q$ of numbers modulo $Q$. This will appear as an intermediate problem in our algorithm in Section~\ref{sec:BKW} and our worst-case to average-case reduction in Section~\ref{sec:wcac}.

\begin{definition}[Average-case Modular $k$-SUM]
For integers $m, k \geq 2$ and integer modulus $Q \ge 2$, the average-case {$\kSUMG{k}{\Z_Q}{m}$} problem is the search problem defined as follows.  The input is  $a_1,\ldots, a_m \sim$ {$\Z_Q$} chosen uniformly and independently at random, and the goal is to
find $k$ distinct elements $a_{i_1},\ldots, a_{i_k}$ with $a_{i_1} + \cdots + a_{i_k} = 0$ {$\pmod{Q}$}.
\end{definition}

We highlight the distinction in our notation for the two problems. The former (non-modular version) is denoted $\kSUMG{k}{u}{m}$ (the first parameter is the bound $u$ on the absolute value of the elements), whereas the latter is denoted {$\kSUMG{k}{\Z_Q}{m}$}  (the first parameter indicates the group on which the problem is defined). The second parameter always refers to the number of elements in the instance.

We now show that the modular problem is total when $\binom{m}{k} \gtrsim Q$ and is unlikely to have a solution when $\binom{m}{k} \lesssim Q$.

\begin{lemma}\label{lem:totalmod}
    If $a_1,\ldots, a_m \sim \Z_Q$ are sampled uniformly at random, and $E_k$ is the event that there exist distinct indices $i_1,\ldots, i_k$ with $a_{i_1} + \cdots +  a_{i_k} = 0 \pmod{Q}$, then
    \[
        1-Q/\binom{m}{k} \leq \Pr[E_k] \leq \binom{m}{k}/Q
        \; .
    \]
\end{lemma}
\begin{proof}
    Notice that for fixed indices $i_1,\ldots, i_k$, the probability that $a_{i_1} + \cdots + a_{i_k} = 0$ is exactly $1/Q$. The upper bound then follows from a union bound over all $\binom{m}{k}$ $k$-tuples of indices. Furthermore, notice that $i_1,\ldots, i_k$ and $j_1,\ldots, j_k$, the event that $a_{i_1} + \cdots + a_{i_k} = 0 \pmod{Q}$ is independent of the event that $a_{j_1} + \cdots + a_{j_k} = 0 \pmod{Q}$ as long as $\{i_1,\ldots, i_k\} \neq \{j_1,\ldots, j_k\}$. The lower bound then follows from Chebyshev's inequality.
\end{proof}

\begin{lemma}
    If $a_1,\ldots, a_m \sim [-u,u]$ are sampled uniformly at random, and $E_k$ is the event that there exist distinct indices $i_1,\ldots, i_k$ with $a_{i_1} + \cdots +  a_{i_k} = 0$, then
    \[
        1- e^{-\alpha} \leq \Pr[E_k] \leq \binom{m}{k}/(2u+1)
        \; ,
    \]
    where 
    \[
        \alpha := \frac{1}{4k+2} \cdot \Big\lfloor \frac{m}{k(20u+10)^{1/k}} \Big\rfloor \approx m/(k^2 u^{1/k})
        \; .
    \]
\end{lemma}
\begin{proof}
    The upper bound follows immediately from the upper bound in Lemma~\ref{lem:totalmod} together with the observation that elements that sum to zero over the integers must sum to zero modulo $Q := 2u+1$ as well.
    
    Let $m' := k (10Q)^{1/k}$. Let $E_k'$ be the event that there exist distinct indices $i_1,\ldots, i_k \leq m'$ with  $a_{i_1} + \cdots +  a_{i_k} = 0$. Notice that
    \[
        \Pr[E_k] \geq 1- (1-\Pr[E_k'])^{\floor{m/m'}} \geq 1 - \exp(-\floor{m/m'}\Pr[E_k'])
        \; .
    \]
    So, it suffices to show that 
    \[
        \Pr[E_k'] \geq \frac{1-Q/\binom{m'}{k}}{2k+1} \geq \frac{1}{4k+2}
        \; .
    \]
    
    By Lemma~\ref{lem:totalmod}, we know that with probability at least $1-Q/\binom{m'}{k}$, there exists a $k$-SUM that sums to zero modulo $Q$ in the first $m'$ elements. I.e., $a_{i_1} + \cdots + a_{i_k} = \ell Q$ for some $\ell \in \{-k,-k+1,\ldots, k-1,k\}$ and $i_1,\ldots, i_k \leq m'$. We wish to argue that $\ell = 0$ is at least as likely as $\ell = i$ for any $i$.

    Let $p(k',s) := \Pr[a_1 + \cdots + a_{k'} = s]$ for integers $k',s$. Notice that for $s \geq 0$, we have 
    \begin{align*}
        p(k'+1,s) - p(k'+1,s+1) 
            &= \big(p(k',-(s+u)) - p(k', s+u+1)\big)/(2u+1) \\
            &= \big(p(k', s+u) - p(k',s+u+1)\big)/(2u+1)
        \; .
    \end{align*}
    It then follows from a simple induction argument that $p(k,s+1) \leq p(k,s)$. In particular, $p(k,\ell Q) \leq p(k,0)$ for any $\ell$. 
    Therefore, letting $E_{k,Q}'$ be the event that the first $m'$ elements contain a $k$-SUM modulo $Q$, we have 
    \begin{align*}
        \Pr[E_k'] 
            &\geq \Pr[E_{k,Q}'] \cdot \Pr[a_{i_1} + \cdots + a_{i_k} = 0 \ | \ a_{i_1} + \cdots + a_{i_k} = 0 \pmod{Q}]\\
            &\geq \frac{\Pr[E_{k,Q}']}{2k+1}
            \; .
    \end{align*}
    Finally, by Lemma~\ref{lem:totalmod}, we have
    \[
        \Pr[E_{k,Q}'] \geq 1-Q/\binom{m'}{k} \geq 1/2
        \; ,
    \]
    as needed.
\end{proof}

\subsection{From \texorpdfstring{$k$}{k}-SUM to Modular \texorpdfstring{$k$}{k}-SUM and Back}

We first show that an algorithm for the modular $k$-SUM problem gives us an algorithm for the $k$-SUM problem. A consequence of this is  that when we describe the algorithm for $k$-SUM in Section~\ref{sec:BKW}, we will focus on the modular variant.

\begin{lemma}\label{lem:Qimpliesu}
Let $u$ be a positive integer and let $Q=2u+1$. If there is an algorithm for the $\kSUMG{k}{\Z_Q}{m}$ that runs in time $T$ and succeeds with probability $p$, then there is an algorithm for $\kSUMG{2k}{u}{2m}$ that runs in time $O(T)$ and succeeds with probability at least $p^2/k$.
\end{lemma}

\begin{proof}
 Let $\mathcal{A}$ be the purported algorithm for $\kSUMG{k}{\Z_Q}{m}$.
 The algorithm for $\kSUMG{2k}{u}{2m}$  receives $2m$ integers $a_1,\ldots,a_{2m}$ in the range $[-u, u]$ and works as follows. We use the natural embedding to associate elements in $\Z_Q$ with elements in $[-u, u]$, so we may think of $a_1,\ldots,a_{2m}$ also as elements in $\Z_Q$ (simply by considering their coset modulo $Q$).
 
 \begin{itemize}
     \item Run $\mathcal{A}$ on $a_1,\ldots,a_m$ to obtain a $k$-subset  $S_1$. If $\mathcal{A}$ does not succeed, then fail.
     \item Run $\mathcal{A}$ on $-a_{m+1},\ldots,-a_{2m}$ to obtain a $k$-subset $S_2$. If $\mathcal{A}$ does not succeed, then fail.
     \item If $\sum_{i\in S_1} a_i = -\sum_{i \in S_2} a_i$, output $S_1\cup S_2$ as the $2k$-subset. Fail otherwise.
 \end{itemize}
 
It is clear that the run-time is $O(T)$ and that if the algorithm does not fail then it indeed outputs a valid $2k$-sum. It suffices to bound the probability that the algorithm succeeds.

Since the first two steps run $\mathcal{A}$ on independent and identically distributed input, we can deduce that the probability that both succeed is $p^2$, and in the case that both succeed, their output satisfies
$$ \sum_{i\in S_1} a_i = \alpha_1 Q \hspace{.1in} \mbox{and} \hspace{.1in} \sum_{i\in S_1} a_i = \alpha_2 Q $$
for some integers $\alpha_1,\alpha_2 \in (-k/2,k/2)$, which are independent and identically distributed random variables. The probability that $\alpha_1 = \alpha_2$ is therefore at least $1/k$, since the collision probability of a random variable is bounded by the inverse of its support size. If this happens then, $\sum_{i\in S_1} a_i = \sum_{i\in S_2} a_i$ and the algorithm succeeds. Thus, we conclude that our algorithm succeeds with probability at least $p^2/k$.
\end{proof}

Finally, we show a proof in the other direction. Namely, that an algorithm for the $k$-SUM problem gives us an algorithm for the modular $k$-SUM problem. We will use this when we describe the worst-case to average-case reduction in Section~\ref{sec:wcac}.

\begin{lemma}\label{lem:mod-to-z}
 For $m \geq k\cdot u^{2/k}$, if there is an algorithm for $\kSUMG{k}{u}{m}$ that runs in time $T$ and succeeds with probability $p$, then there is an algorithm for $\kSUMG{k}{\Z_{2u+1}}{m}$ that runs in time $T$ and succeeds with probability $p$. 
\end{lemma}

\begin{proof}
Let $\mathcal{A}$ be the purported algorithm for $\kSUMG{k}{u}{m}$.
 The algorithm for $\kSUMG{k}{\Z_{2u+1}}{m}$   receives $m$ integers $a_1,\ldots,a_{2m} \in \mathbb{Z}_{2u+1}$ and works as follows. 
 
 As before, identify $\Z_{2u+1}$ with the interval $[-u,u]$ and run $\mathcal{A}$ on $a_1,\ldots,a_m$. By Lemma~\ref{lem:mod-to-z}, since $m \geq k\cdot u^{2/k}$, there is a $k$-subset $S$ such that $\sum_{i\in S} a_i = 0$. In particular, $\sum_{i\in S} a_i = 0 \pmod{Q}$, so it is a modular $k$-SUM solution as well.
\end{proof}

\section{The  \texorpdfstring{$u^{O(1/\log k)}$-time}{faster} Algorithm for Average-case \texorpdfstring{$k$}{k}-SUM}
\label{sec:BKW}

In this section, we describe a variant of the Blum-Kalai-Wasserman algorithm~\cite{BKW03} for the average-case $k$-SUM problem that runs in time $u^{O(1/\log k)}$.

\begin{theorem}
    There is a $\widetilde{O}(2^\ell q^2)$-time algorithm that solves average-case $\kSUMG{2^\ell}{\Z_{q^\ell}}{m}$ for $m = \widetilde{\Theta}(2^\ell q^2)$.
\end{theorem}
\begin{proof}
    On input $a_1,\ldots, a_m \in \Z_{q^\ell}$ with $m := 1000\ell^2 q^2 2^\ell \log q = \widetilde{\Theta}(2^\ell q^2)$, the algorithm behaves as follows. Let $L_1 := (a_1,\ldots, a_m)$.
    For $i = 1,\ldots, \ell$, the algorithm groups the elements in $L_i$ according to their value modulo $q^i$. It then greedily groups them into $m_{i+1}$ disjoint points $(a,b)$ with $a + b = 0 \bmod q^i$. It sets $L_{i+1}$ to be the list of sums of these pairs (and records the indices of the $2^i$ input elements that sum to $a+b$). If at any point the algorithm fails to find such pairs, it simply fails; otherwise, the algorithm outputs the elements $a_{i_1}, \cdots , a_{i_{2^\ell}}$ satisfying $\sum a_{i_j} = 0 \bmod q^\ell$ found in the last step.
    
    The running time of the algorithm is clearly $\poly(\ell, \log q, \log m) m$ as claimed. To prove correctness, we need to show that at each step the algorithm is likely to succeed in populating the list $L_{i+1}$ with at least $m_i := (\ell^2-i^2)/\ell^2 \cdot m/2^{i-1}$ elements, since clearly the algorithm outputs a valid $2^\ell$-SUM in this case.
    
    Suppose that the algorithm succeeds up to the point where it populates $L_i$. Let $L_i = (b_1,\ldots, b_{m_i})$, and $b_i' := (b_i/q^{i-1}) \bmod q$, where the division by $q^{i-1}$ is possible because $b_i = 0 \bmod q^{i-1}$ by assumption. Notice that the $b_i'$ are independent and uniformly random. For $j \in \Z_q$, let $c_j := |\{i \ : \ b_i' = j \bmod q\}|$. Notice that the algorithm successfully populates $L_{i+1}$ if and only if
    \[
        \sum_{j\in\Z_q} \min\{ c_j, c_{-j} \}/2 \geq m_{i+1}
        \; .
    \]
    By the Chernoff-Hoeffding bound, we have that 
    \[
        \Pr\big[c_j < m_i/q - 10\sqrt{m_i \log m_i}\big] \leq 1/m_i^2
    \]
    It follows that
    \[
        \sum_j \min\{ c_j, c_{-j} \}/2 \geq q \min \{ c_j\}/2 \geq m_i/2 - 5q \sqrt{m_i \log m_i} \geq m_{i+1}
    \]
    except with probability at most $1/m_i$. By union bound, we see that the algorithm succeeds in populating every list except with probability at most $\sum 1/m_i \ll 1/10$, as needed.
\end{proof}

Combining this with Lemma~\ref{lem:Qimpliesu} (the reduction from $k$-SUM to modular $k$-SUM), we obtain the following corollary.

\begin{corollary}
    For $u = (q^\ell-1)/2$ for odd $q$ and $k = 2^{\ell+1}$, there is a $u^{O(1/\log k)}$-time algorithm for $\kSUMG{k}{u}{m}$ for $m  = u^{\Theta(1/\log k)}$.
\end{corollary}

\section{From Worst-case Lattice Problems to Average-case \texorpdfstring{$k$}{k}-SUM}
\label{sec:simple}\label{sec:wcac}

In this section, we describe our main result, namely a worst-case to average-case reduction for $k$-SUM. We state the theorem below.

\begin{theorem}\label{thm:mainthm-red}
  Let $k,m,u,n$ be positive integers, and $0 < \eps < \eps'$ where $$u = k^{2(1+\eps')cn/\eps'} \mbox{  and  }
  m = u^{\eps /(2\log k)}~$$
  for some universal constant $c > 0$.
  If there is an algorithm for average-case $\kSUMG{k}{u}{m}$ that runs in time $T_{\mathsf{kSUM}} = T_{\mathsf{kSUM}}(k,u,m)$,  then there is an algorithm for the worst-case  $n^{1+\eps'}$-approximate shortest independent vectors problem (SIVP) that runs in time $2^{O(\eps n/\eps' + \log n)} \cdot T_{\mathsf{kSUM}}$.
\end{theorem}

When we say that a $k$-SUM algorithm succeeds, we mean that it outputs a $k$-subset of the input that sums to $0$ with probability $1-\delta$ for some tiny $\delta$. This can be achieved starting from an algorithm that succeeds with (some small) probability $p$ by repeating, at the expense of a multiplicative factor of $1/p\cdot \log(1/\delta)$ in the run-time. We ignore such issues for this exposition, and assume that the algorithm outputs a $k$-sum with probability $1-\delta$ for a tiny $\delta$.

Before we proceed to the proof, a few remarks on the parameters of Theorem~\ref{thm:mainthm-red} are in order. First, note that the parameter settings imply that $m^k \gg u$, therefore putting us in the total regime of parameters for $k$-SUM. Secondly, setting $\epsilon'=100$ (say), we get the following consequence: if there is a $k$-SUM algorithm that, on input $m = u^{\epsilon/(2\log k)}$ numbers, runs in time roughly $m$, then we have an $n^{101}$-approximate SIVP algorithm that runs in time $\approx 2^{\epsilon c n}$. Now, $\epsilon$ is the ``knob'' that one can turn to make the SIVP algorithm run faster, assuming a correspondingly fast $k$-SUM algorithm that works with a correspondingly smaller instance.

\begin{proof}
  The theorem follows from the following observations:
  \begin{itemize}
  \item
  First, by Theorem~\ref{thm:GINX}, there is a reduction from $\widetilde{O}(\sqrt{n\log m'}\cdot \beta)$-approximate SIVP to $\SIS(m',Q,\beta)$, where we take $m' := \ceil{k^{10cn/(k \eps')} n^{10}}$. The reduction produces SIS instances over $\Z_Q$ where $Q \geq (\beta n)^{cn}$ for some constant $c$, and works as long as the SIS algorithm produces solutions of $\ell_1$ norm at most $\beta$. If the SIS algorithm runs in time $T_{\SIS} = T_{\SIS}(m',Q,\beta)$, the SIVP algorithm runs in time $(m'+T_{\SIS})\cdot \mathsf{poly}(n)$. We take $\beta := n^{\eps'}$.

  \item
  Second, as our main technical contribution, we show in Lemma~\ref{lem:main} how to reduce SIS to $k$-SUM. Note that Theorem~\ref{thm:GINX} gives us the freedom to pick $Q$, as long as it is sufficiently large. We will set $Q = q^r$ where $r = \floor{\eps' \log n/(2 \log k)}$ for a prime $q \approx u \approx (\beta n)^{cn/r} \approx k^{2(1+\eps') cn/\eps'}$.

  Now, Lemma~\ref{lem:main} (with $k = t$) shows a reduction from $\SIS(m',Q,\beta)$ to  $\kSUMG{2k}{\Z_q}{m}$ (provided that $m' \gg q^{1/k} k^{4r/k} m^{4/k}$, which holds in this case).
  The reduction produces a SIS solution with $\ell_1$ norm bounded by $k^{2r} \leq \beta$.

  The running time of the resulting algorithm is $$rm'(m \cdot \poly(k,\log q) + 10T_{\mathsf{kSUM}}) \approx 2^{O(\eps n/\eps' + \log n)} \cdot T_{\mathsf{kSUM}}~.$$

  \item
  Finally, by Lemma~\ref{lem:mod-to-z}, we know that modular $2k$-SUM over $\Z_q$ can be reduced to $k$-SUM over the integers in the interval $[-u,u]$ for $u \approx q$ with essentially no overhead.
  \end{itemize}
  This finishes the proof.
\end{proof}

The following lemma shows our main reduction from SIS to $k$-SUM. In particular, taking $k = t$, $m' \gg (m^4q)^{1/k} \cdot (10k)^{4r}$ (so that $\delta$ is small), and $m \gg q^{1/k}$ (so that $\kSUMG{k}{\Z_q}{m}$ is total) gives a roughly $rmm'$-time reduction from SIS over $\Z_{q^r}$ to $k$-SUM over $\Z_q$ with high success probability.

\begin{lemma}
\label{lem:main}
Let $m,m', k,r,t$ be positive integers and $q > (tk)^r$ a prime, and let $Q = q^{r}$. If there is an algorithm that solves (average-case) $\kSUMG{k}{\Z_q}{m}$ in time $T$ with success probability $p$, then there is an algorithm that solves $\SIS(m',Q,\beta)$ in time $r\cdot m' (m \cdot  \poly(k,t,\log q)+10T)/p$ with success probability at least $1 - \delta$ and produces a solution with $\ell_1$ norm $\beta \leq (tk)^{r}$, where
  \[
    \delta := \frac{100 rm (m')^2}{p} \cdot  \frac{q^{1/4}}{(m'/(10tk)^{2r+1})^{t/4}}
    \; .
  \]
\end{lemma}
\begin{proof}
    At a high level, the idea is to run a variant of the Blum-Kalai-Wasserman~\cite{BKW03} algorithm where in each iteration, we call a $k$-SUM oracle. In particular, on input $a_1,\ldots,a_{m'}$, the algorithm operates as follows.

  \begin{itemize}
  \item
  In the beginning of the $i^{\mathsf{th}}$ iteration, the algorithm starts with a sequence of $$m_i := \ceil{m'/(10 t^2 k^2)^{i-1}}$$ numbers $a_{i,1},\ldots,a_{i,m_i}$. The invariant is that $a_{i,j} = 0 \pmod{q^{i-1}}$ for all $j$. It then generates disjoint $S_{i,1},\ldots, S_{i,m_{i+1}} \subseteq [m_i]$ such that $|S_{i,\ell}| \leq kt$ and $\sum_{j \in S_{i,\ell}} a_{i,j} =0 \pmod{q^i}$, in a way that we will describe below.

  As the base case, for $i=1$, $a_{1,j} = a_j$, the input itself, and the invariant is vacuous.
  \item In the $i^{\mathsf{th}}$ iteration, we apply the re-randomization lemma (Corollary~\ref{cor:LHL_A}), computing subsets of $t$ randomly chosen elements from $a_{i,1},\ldots, a_{i,m_i}$, to generate $m_i^* := 10m \ceil{m_{i+1}/p}$ numbers $c_{i,1},\ldots,c_{i,m_i^*}$.
  \item Let $d_{i,j} = c_{i,j}/q^{i-1} \pmod{q} \in \Z_q$. Note that this is well-defined because each $c_{i,j} = 0 \pmod{q^{i-1}}$.
  \item Divide the $d_{i,j}$ into $10\ceil{m_{i+1}/p}$ disjoint blocks of $m$ elements each, set $\ell = 1$. For each block, feed the block to the $k$-SUM algorithm to obtain $d_{i,j_1},\ldots, d_{i,j_k}$.  This yields corresponding subsets $S_{1}^*,\ldots, S_{k}^* \in \binom{[m_i]}{t}$ such that $\sum_{j \in S_x^*} a_{i,j}/q^{i-1} = d_{i,j_x} \pmod{q}$. If $d_{i,j_1} + \cdots + d_{i,j_k} = 0 \pmod{q}$ and the sets $S_1^*,\ldots S_{k}^*,S_{i,1},\ldots, S_{i,\ell-1}$ are pairwise disjoint, then set $S_{i,\ell} := \bigcup S_x^*$ and increment $\ell$.
  \item If $\ell \leq m_{i+1}$, the algorithm fails. Otherwise, take $a_{i+1,\ell} := \sum_{j \in S_{i,\ell}} a_{i,j}$ for $\ell = 1,\ldots, m_{i+1}$.
  \item At the end of the $r^{th}$ iteration we obtain a $(kt)^{r}$-subset of the $a_1,\ldots,a_{m'}$ that sums to $0 \pmod{Q}$.
  \end{itemize}
  We now analyze the correctness, run-time and the quality of output of this reduction.

  The reduction calls the $k$-SUM oracle $\sum m_i^*/m \leq 20rm'/p$ times. The rest of the operations take $r mm' \poly(k,t,\log q)/p$ time for a total of $r\cdot m'( m \poly(k,t,\log q)+10T)/p$ time, as claimed. Furthermore, the $\ell_1$ norm of the solution is $\beta \leq (tk)^r$, as claimed.

  Finally, we show that the algorithm succeeds with the claimed probability. Since the sets $S_{i,\ell}$ are disjoint and do not depend on $a_{i,j} - (a_{i,j} \bmod q^i)$, it follows from a simple induction argument that at each step the $a_{i,j}$ are uniformly random and independent elements from $q^{i-1}\Z/q^r\Z$.
  Therefore, by Corollary~\ref{cor:LHL_A}, the statistical distance of the collection of all $d_{i,j}$ (for a given $i$) from uniformly random variables that are independent of the $a_{i,j}$ is $\delta_i \leq (m_i^*+1)\cdot (\frac{qt^t}{m_i^t})^{1/4}$. In total, the statistical distance of all samples from uniform is then at most $\sum \delta_i < \delta/3$ for our choice of parameters. So, up to statistical distance $\delta/3$, we can treat the $d_{i,j}$ as uniformly random and independent elements.

  It remains to show that, assuming that the $d_{i,j}$ are uniformly random and independent, then we will find \emph{disjoint} sets $S_{i,1},\ldots, S_{i,m_{i+1}}$ with $\sum_{j \in S_{i,\ell}} d_{i,j} = 0 \pmod{q}$ at each step except with probability at most $2\delta/3$.
  Let $\vec{b}_{i} := (a_{i,1}/q^{i-1} \bmod q,\ldots, a_{i,m_i}/q^{i-1} \bmod q)$. By Corollary~\ref{cor:hitting_prob}, we have
  \[
    p_{\vec{b}_i, \vec{d}_i, I,J,t} \leq 10t^2 k^2 m_{i+1}/m_i \leq 1/2
    \]
    for all $I \in \binom{[m_i]}{\leq v}$ and $J \in \binom{[m_i^*]}{\leq v}$ except with probability at most $10 m_i m_i^* q^{1/4}/\binom{m_i-1}{t-1}^{1/4} < \delta/3$ for $v := tk m_{i+1} \geq |T|$, $v' := k$, and $\eps := 1/2$.

  So, suppose this holds.
  Notice that $\Pr[d_{i,j_1} + \cdots + d_{i,j_k} = 0 \pmod{q}] = p$ by definition. And, conditioned on $d_{i,j_1},\ldots, d_{i,j_k}$, the $S_{x}^*$ are independent and uniformly random subject to the constraint that $\sum_{j \in S_x^*} a_{i,j}/q^{i-1} = d_{i,j_x} \pmod{q}$. Therefore, the probability that $S_1^*,\ldots, S_k^*, I := S_{i,1} \cup \cdots \cup S_{i,\ell-1}$ are pairwise disjoint in this case is exactly
  \[
    p_{\vec{b}_i, \vec{d}_i, I,J,t} \leq 1/2
    \; ,
  \]
  where $J := \{j_1,\ldots, j_k\}$. So, each time we call the oracle, we increment $\ell$ with probability at least $p/2$. It follows from the Chernoff-Hoeffding bound that we increment $\ell$ at least $m_{i+1}$ times except with probability at most $e^{-m_{i+1}/100}\ll \delta/3$.

  Putting everything together, we see that the algorithm fails with probability at most $\delta$, as claimed.
\end{proof}

\newcommand{\etalchar}[1]{$^{#1}$}

\appendix

\section{Total \texorpdfstring{$k$}{k}-SUM and Computational Geometry} 
\label{apx:geometric}

Here, we show that one of the main results in~\cite{GOClassProblems95,GOClassProblems12} can be extended meaningfully to our setting, i.e., to the case of search $k$-sum over $\Z_Q$ with $\binom{m}{k} \gg Q$. (In~\cite{GOClassProblems95,GOClassProblems12}, Gajentaan and Overmars only considered decisional $3$-SUM.) Specifically, we will reduce the following problem to $k$-SUM in this regime.

\begin{definition}
For $d \geq 1$ and $Q,m \geq 2$ with $Q$ prime, the \emph{$(Q,m,d)$-Plane problem} is the following search problem. The input is $\vec{a}_1,\ldots, \vec{a}_m \in \Z_{Q}^{d+1}$. The goal is to find distinct $\vec{a}_{i_1},\ldots, \vec{a}_{i_{d+2}}$ that lie in a $d$-dimensional affine hyperplane over the field $\Z_{Q}$. (In other words, $\vec{a}_{i_{d+1}} - \vec{a}_{i_{d+2}}$ can be written as a linear combination of $ \vec{a}_{i_1} - \vec{a}_{i_{d+2}},\ldots, \vec{a}_{i_{d}} - \vec{a}_{i_{d+2}}$ over $\Z_Q$.)
\end{definition}

\begin{lemma}
    For $d \geq 1$ and $Q,m \geq 2$ with $Q$ prime, there is a reduction from $\kSUMG{(d+2)}{\Z_Q}{m}$ to $(Q,m,d)$-Plane.
\end{lemma}
\begin{proof}
    Let $f_d : \Z_Q \to \Z_{Q}^{d+1}$ be the map $f_d({a}) := ({a},{a}^2,{a}^3,\ldots, {a}^{d},{a}^{d+2})$. 
    E.g., $f_1({a}) = ({a},{a}^3)$, $f_2({a}) = ({a},{a}^2,{a}^4)$, etc. On input ${a}_1,\ldots, {a}_m \in \Z_Q$, the reduction simply calls its $(Q,m,d)$-Plane oracle on $f_d({a}_1),\ldots, f_d({a}_m) \in \F_{Q}^{d+1}$, receiving as output distinct indices $i_1,\ldots, i_{d+2}$ such that $f_{d}({a}_{i_1}),\ldots, f_d({a}_{i_{d+2}})$ lie in a $d$-dimensional affine hyperplane (assuming that such indices exist). The reduction simply outputs these indices, i.e., it claims that ${a}_{i_1} + \cdots + {a}_{i_{d+2}} = 0 \bmod q$.

    Notice that $d+2$ points $\vec{b}_{1},\ldots, \vec{b}_{d+2} \in \Z_{Q}^{d+1}$ lie in a $d$-dimensional affine hyperplane if and only if the matrix
    $(
        \vec{b}_1 - \vec{b}_{d+2}, \vec{b}_2 - \vec{b}_{d+2}, \ldots, \vec{b}_{d+1} - \vec{b}_{d+2})\in \Z_{Q}^{(d+1) \times (d+1)}
    $
    has determinant zero. (Here, we have used the fact that $\Z_{Q}$ is a field.) So, we consider the matrix
    \[
        \vec{M} := \vec{M}(b_1,\ldots, b_{d+2}) := \begin{pmatrix}
            b_1 - b_{d+2} & b_2 - b_{d+2} & \cdots &b_{d+1} - b_{d+2}\\
            b_1^2 - b_{d+2}^2 &b_2^2 - b_{d+2}^2 &\cdots &b_{d+1}^2 - b_{d+2}^2\\
            \vdots & \vdots &\ddots &\vdots\\
            b_1^{d} - b_{d+2}^{d} & b_2^{d} - b_{d+2}^{d} &\cdots & b_{d+1}^{d} - b_{d+2}^{d}\\
            b_1^{d+2} - b_{d+2}^{d+2} & b_2^{d+2} - b_{d+2}^{d+2} &\cdots & b_{d+1}^{d+2} - b_{d+2}^{d+2}
        \end{pmatrix}
        \in \F_{Q}^{(d+1) \times (d+1)}
        \; .
    \]
    We claim that
    \[
    \det(\vec{M}) = (-1)^d(b_1 + \cdots + b_{d+2}) \cdot \prod_{i < j} (b_j - b_i)
    \; ,
    \]
    The result then follows, since this is zero if and only if $b_i = b_j$ for some $i \neq j$ or $b_1 + \cdots + b_{d+2} = 0$. Since by definition the $(Q,m,d)$-Plane oracle only outputs distinct vectors on a hyperplane, this means that its output must correspond to distinct elements with ${a}_{i_1} + \cdots + {a}_{i_{d+2}} = 0 \bmod Q$.

    To prove that the determinant has the appropriate form, we first notice that without loss of generality we may take $b_{d+2} = 0$. Next, we define
    \[
        \vec{M}' := \vec{M}'(b_1,\ldots, b_{d+1}) := \begin{pmatrix}
            b_1 & b_2 & \cdots &b_{d+1}\\
            b_1^2 &b_2^2 &\cdots &b_{d+1}^2\\
            \vdots & \vdots &\ddots &\vdots\\
            b_1^{d} & b_2^{d} &\cdots & b_{d+1}^{d}\\
            b_1^{d+1} & b_2^{d+1} &\cdots & b_{d+1}^{d+1}
            \; .
        \end{pmatrix} \in \F_{Q}^{(d+1) \times (d+1)}
    \]
    This is just a Vandermonde matrix with columns scaled up by $b_i$. So, its determinant is a scaling of the Vandermonde determinant,
    \[
        \det(\vec{M}') =  b_1 \cdots b_{d+1} \cdot \prod_{i < j} (b_j - b_i)
        \; .
    \]

    Finally, we recall Cramer's rule, which in particular tells us that
    \[
        \det(\vec{M}) = p_{d+1} \det(\vec{M}')
    \]
    for the unique $\vec{p} := (p_1,\ldots, p_{d+1}) \in \Z_Q^{d+1}$ satisfying $\vec{p}^T\vec{M}' = (b_1^{d+2},\ldots, b_{d+1}^{d+2})$. I.e., the coordinates of $\vec{p}$ form the polynomial $p(x) := p_1 + p_2 x + \cdots + p_{d+1}x^d$ such that $p(b_i) = b_i^{d+1}$. The result follows by noting that $p_i = (-1)^{i-1} \sum_{S \in \binom{[d+1]}{d+2-i}} \prod_{j \in S} b_j$. In particular, $p_{d+1} = (-1)^d(b_1 + \cdots + b_{d+1})$, as needed.
\end{proof}

\end{document}